\documentclass[12pt,a4paper]{article}

\usepackage{amsfonts,amsmath,amssymb,amsthm}
\usepackage{latexsym,amscd}
\usepackage{amsbsy}
\usepackage{CJK}
\usepackage{indentfirst,mathrsfs}
\usepackage{amsfonts,amsmath,amssymb,amsthm}
\usepackage{latexsym,amscd}
\usepackage{amsbsy}
\usepackage{color}
\usepackage{multirow}
\usepackage{makecell}
\usepackage{float}
\usepackage{cases}
\usepackage{booktabs}
\usepackage{xcolor}
\usepackage{cite}
\usepackage{longtable}

\usepackage{amsmath,amssymb}
\usepackage{booktabs}
\usepackage{array}
\usepackage{makecell}
\usepackage{multirow}
\usepackage{float}

\allowdisplaybreaks

\numberwithin{equation}{section}

\newtheorem{thm}{Theorem}
\newtheorem{lem}{Lemma}
\newtheorem{cor}{Corollary}
\newtheorem{prop}{Proposition}

\newtheorem{example}{Example}
\newtheorem{remark}{Remark}

\newcommand{\Tr}{{\rm {Tr}}}
\newcommand{\F}{\mathbb{F}}

\begin{document}

\title{A class of Niho type four-valued cross-correlation functions between sequences of different periods}
	\author{Yao Yao \thanks{Y. Yao, Y. Wang, and X. Zeng are with Hubei Key Laboratory of Applied Mathematics, Faculty of Mathematics and Statistics, Hubei
		University, Wuhan, 430062, China. Email: yyao1966@163.com (Y. Yao), yufeiwang11@163.com (Y. Wang), xiangyongzeng@aliyun.com (X. Zeng)}, Yufei Wang,
Yongbo Xia\thanks{Y. Xia is with  the Department of Mathematics and Statistics, South-Central Minzu University,
Wuhan 430074, China. Email: xia@mail.scuec.edu.cn},
Nian Li\thanks{N. Li is with Hubei Provincial Engineering Research Center of Intelligent Connected Vehicle Network Security, School of Cyber Science and Technology, Hubei
	University, Wuhan, 430062, China. Email: nian.li@hubu.edu.cn},
Xiangyong Zeng}
	\date{}
	\maketitle
\date{}
\maketitle

\begin{quote}
{\small {\bf Abstract:}
Determining the cross-correlation distributions between $m$-sequences and their decimated sequences is an important problem in sequence design. For two positive integers $n$ and $k$ with $n=2k$, this paper investigates the cross-correlation function between a binary $m$-sequence $\{u(t)\}$ of period $2^{n}-1$ and its $2^{\gcd(s,k)}+1$ short-period decimated sequences $\{u(dt+l)\}$, where the decimation $d=\frac{2^s}{2^s+1}(2^k-1)+1$ satisfies $1<s<n$ and $v_2(s)>v_2(k)$. The cross-correlation function is proved to be four-valued and its value distribution is completely determined. Several previously known binary cross-correlation functions are included in our result as special cases.
}

\end{quote}

\section{Introduction} \label{intro}

Sequences with low correlation are widely used in spread spectrum communications, code-division multiple access (CDMA), radar, and stream cipher systems \cite{Helle2, GolombGong1}.
As the most well-known sequences possessing ideal out-of-phase autocorrelation, the maximal-length linear-feedback shift-register sequences ($m$-sequences) have been extensively studied \cite{Helle2,GolombGong1}. This property motivates analyzing the cross-correlation between two different $m$-sequences related by $d$-decimation, where $d$ is coprime to their common period. Since the cross-correlation function of any two $m$-sequences that are not cyclic shifts of
each other assumes  at least three
values \cite{Helle1}, considerable research  has focused on characterizing decimations that yield exactly three
cross-correlation values (see \cite{CD, Gold1,Kasami,Helle1,DobbWelch1,Canteaut1,Hollmann1} and references therein).

Another important problem in the theory of $m$-sequences is to investigate the four-valued cross-correlation functions between two different $m$-sequences of the same period. Niho studied two four-valued cases, namely the decimations $d=2^{\frac{n}{2}+1}-1$ and $(2^{\frac{n}{2}}+1)(2^{\frac{n}{4}}-1)+2$ with  $n\equiv0\pmod4$ \cite[Theorems 3.6 and 3.7]{Niho1}.
Dobbertin considered the decimation $d=\sum_{i=0}^{\frac{n}{2}}2^{it}$,
with $n\equiv0\pmod4$, $0<t<n$, and $\gcd(t,n)=1$ \cite{Dobb1}. When $t=1$,
this decimation reduces to Niho's decimation $d=2^{\frac{n}{2}+1}-1$. Helleseth and Rosendahl proposed the decimation $d=\frac{1}{2^t-1}(2^{n}+2^{t+1}-2^{\frac{n}{2}+1}-1)$ with $t>0$, $n\equiv0\pmod4$, and $4t\mid n$
\cite{Helle3}. This family yields four-valued cross-correlation
functions, and the cases $t=1$ and $t=n/4$ are equivalent to
the two Niho decimations above, respectively.  Building on these results, Dobbertin et al. further investigated the Niho-type decimations for even $n$ and integers
$s$ satisfying $0<s<n/2$ and $v_2(s)<v_2(n/2)$:
\begin{equation}\label{d12}
d_1
=\frac{2^s}{2^s+1}(2^{\frac{n}{2}}-1)+1
\quad\text{and}\quad
d_2=\frac{2^s}{2^s-1}(2^{\frac{n}{2}}-1)+1,
\end{equation}
where $v_2(N)$ denotes the exponent of $2$ in the prime factor decomposition of a positive integer $N$, and $\frac{1}{2^s\pm1}$ denotes the multiplicative inverse of $2^s\pm1$ modulo $2^{n/2}+1$ \cite{Dobb2}. As noted by Dobbertin et al. \cite{Dobb2}, the decimations in \eqref{d12} include the four-valued
Niho-type cases known at that time and also yield new
four-valued cross-correlation functions. Some four-valued cross-correlation functions in nonbinary cases were also studied in \cite{Helle1,Xia1,Ranto1}.

Let $p$ be a prime and $n$ a positive integer. For an $m$-sequence $\{u(t)\}$ of period $p^n-1$, when the decimation $d$ is not
coprime to $p^n-1$, it gives rise to a family of short-period
decimated sequences $\{u(dt+l)\}$, $l=0,1,\ldots,g-1$, where $g=\gcd(d,p^n-1)>1$. The cross-correlation between an $m$-sequence and these
short-period sequences is also of considerable interest and has been studied extensively over the past several
decades \cite{Helle6,Kumar1,Muller1,Ness1,Ness2,Hu, Xia2010,Seo1,Seo2,Xia0}. For example, M$\ddot{\rm u}$ller characterized the cross-correlation function $C_d(t)$ between a ternary $m$-sequence of period $3^m-1$ and its $d$-decimated sequence of period $\frac{3^m-1}{2}$, where $m$ is odd and $d=\frac{3^m+1}{4}+\frac{3^m-1}{2}$. He established an upper bound on $|1+C_d(t)|$ in \cite{Muller1}. This result was subsequently generalized in \cite{Hu} and improved in \cite{Xia2010}. For many decimations $d$, the values and distribution of
the cross-correlation between an $m$-sequence and a short-period $d$-decimated sequence have been
determined \cite{Fan1,Helle6,Kumar1,Luo4,Luo5,Luo1,Ness1,Ness2,Seo1,Seo2,Xia2}. However, these results often concern only a fixed value
of $l$. Complete distributions for every $l=0,1,\ldots,g-1$ are known for only a few decimations $d$;
see \cite{Fan1, Kumar1, Luo4,Luo1,Seo1,Seo2,Xia2}. Among the known cross-correlation functions between sequences
of different periods, three-valued cases can be found in \cite{Helle6,Luo5}, and four-valued cases in
\cite{Helle6,Ness2,Luo4,Fan1,Luo1,Seo2,Xia2}.

The purpose of this paper is to determine the four-valued
cross-correlation distributions between an $m$-sequence
and each of its short-period $d$-decimated sequences.
To this end, we revisit the Niho-type decimations in
\eqref{d12} under parameter conditions for which the
decimation is not coprime to the period.  Let $n=2k$, where $k$ is a positive integer.
As shown in \cite{Dobb2}, for $0<s<k$, there are two
noncoprime cases of \eqref{d12}: $d_1$ with
$v_2(s)>v_2(k)$ and $d_2$ with $v_2(s)=v_2(k)$. We consider the following Niho-type decimation:
\begin{equation}\label{new-d}
	d=\frac{2^s}{2^s+1}\left(2^k-1\right)+1,\,\,1<s<n,\,\,v_2(s)>v_2(k).
\end{equation}
Here we extend the range of $s$ to $1<s<2k$. The condition $v_2(s)>v_2(k)$ ensures that $2^s+1$ is invertible modulo $2^k+1$ (see Lemma \ref{lem3} for details), so $d$ is well defined modulo $2^{n}-1$. This condition also excludes $s=k$. For $1<s<k$, the decimation in \eqref{new-d} gives the
noncoprime $d_1$ branch directly. For $k<s<2k$, the
congruence $2^k\equiv-1\pmod{2^k+1}$ gives
$$
 \frac{2^s}{2^s+1}\equiv
 \frac{-2^{s-k}}{-2^{s-k}+1}\equiv
 \frac{2^{s-k}}{2^{s-k}-1}\pmod{2^k+1}.
$$
Consequently,
$$d\equiv
\frac{2^{s-k}}{2^{s-k}-1}(2^k-1)+1
\pmod{2^n-1}.$$ Since $0<s-k<k$ and $v_2(s-k)=v_2(k)$, this gives
the noncoprime $d_2$ branch with parameter $s-k$. Conversely, suppose that $0<s<k$ and $v_2(s)=v_2(k)$. Then $k<s+k<2k$ and $v_2(s+k)>v_2(k)$. Using $2^{s+k}\equiv-2^s\pmod{2^k+1}$, we obtain $$d_2=\frac{2^s}{2^s-1}\left(2^k-1\right)+1\equiv\frac{2^{s+k}}{2^{s+k}+1}\left(2^k-1\right)+1\pmod{2^n-1}.$$ Thus, every decimation in the noncoprime $d_2$ branch
can be represented by \eqref{new-d} with parameter $s+k$. Therefore, \eqref{new-d} provides a unified representation
of the two noncoprime branches of \eqref{d12}.

In this paper, for the decimation $d$ in \eqref{new-d}, we
investigate the cross-correlation between a binary
$m$-sequence and each of its short-period $d$-decimated
sequences. Each cross-correlation function
is proved to be four-valued and its value
distribution is completely determined. Our method does not require the third or higher moment
identities. We relate a coefficient of the associated
unit-circle equation to one of its roots, and then
establish a double-counting identity for the desired
frequency. To evaluate this identity, we fix a root
and map the remaining roots bijectively to the solutions
of a linearized equation in the trace-one set.
The remaining frequencies follow from the total
frequency and the first two moments.

 Table \ref{tab:table1} summarizes some known Niho-type
decimations that yield four-valued cross-correlation functions.
The previously known noncoprime binary decimations listed
in the table satisfy $\gcd(d,2^n-1)=3$ and yield three
short-period decimated sequences. In contrast, our decimation $d$ in \eqref{new-d} yields
$2^{\gcd(s,k)}+1$ cyclically inequivalent short-period
decimated sequences, a quantity strictly greater than $3$ provided that $\gcd(s,k)>1$. The decimations in \cite{Fan1,Luo1} and the binary case of
\cite{Xia2} can be written in the normalized Niho forms
\begin{equation}\label{forcomp}
d=\frac{2^t}{2^t-1}(2^k-1)+1,\quad
d=\frac15(2^k-1)+1,\quad
d=2(2^k-1)+1,
\end{equation} respectively, where $k>1$ is odd and, in the first case,
$t$ is odd, $1\le t<k$, and $\gcd(t,k)=1$.
These decimations are special cases of \eqref{new-d},
with corresponding parameters $s=k+t$, $s=2k-2$,
and $s=k+1$. Thus, the results for the decimations in \cite{Fan1,Luo1}
and for the binary case of \cite{Xia2} are included
as special cases of our result.

The remainder of this paper is organized as follows.
Section \ref{pre} introduces the necessary preliminaries.
Section \ref{sec3} determines the four-valued cross-correlation
distributions between a binary $m$-sequence and each of
its short-period $d$-decimated sequences.
Finally, Section \ref{con} concludes the study.

\begin{table}[!htb]
	\scriptsize
	\caption{Known Niho-type decimations having four-valued cross-correlation}
	\label{tab:table1}\vspace{0.2cm}

	\centering
	\setlength{\tabcolsep}{1.6pt}
	\renewcommand{\arraystretch}{1.15}
	\begin{tabular}{@{}
            >{\centering\arraybackslash}m{0.08\textwidth}
            >{\raggedright\arraybackslash}m{0.38\textwidth}
            >{\raggedright\arraybackslash}m{0.27\textwidth}
            >{\centering\arraybackslash}m{0.12\textwidth}
            >{\centering\arraybackslash}m{0.08\textwidth}
			@{}}
		\toprule
		\(p\)
		& Exponent \(d\)
		& Parameter conditions
		& \(\gcd(d,p^n-1)\)
		& Ref. \\
		\midrule
		\multirow{5}{*}{\(2\)}
		& $\begin{aligned}
			d_1&=2(2^k-1)+1,\\
			d_2&=(2^k+1)(2^{k/2}-1)+2
		\end{aligned}$
		& $n=2k$, $k$ even
		& $1$
		& \cite{Niho1} \\
		\cmidrule{2-5}
		& $\displaystyle d=\sum_{i=0}^{2k}2^{it}$
		& $n=4k$, $0<t<n$, $\gcd(n,t)=1$
		& $1$
		& \cite{Dobb1} \\
		\cmidrule{2-5}
		& $\displaystyle d=\frac{2^{2k}+2^{t+1}-2^{k+1}-1}{2^t-1}$
		& $n=2k$, $2t\mid k$
		& $1$
		& \cite{Helle3} \\
		\cmidrule{2-5}
	
		& $\begin{aligned}
			d&=s_0(2^k-1)+1,\\
			s&_0\equiv 2^s(2^s\pm1)^{-1}\pmod{2^k+1}
		\end{aligned}$
		& $n=2k$, $0<s<k$, $v_2(s)<v_2(k)$
		& $1$
		& \cite{Dobb2} \\
	
				\midrule
		odd prime
		& $d=2p^k-1$
		& $n=2k$, $p^k\not\equiv2\pmod 3$
		& $1$
		& \cite{Helle1} \\
		\midrule
		\multirow{2}{*}{\(2\)}
		& $d=\rho(2^k-1)+1,\quad 5\rho\equiv1\pmod{2^k+1}$
		& $\begin{aligned}
			n&=2k,\ k\text{ odd};\\
			l&=0,1,2
		\end{aligned}$
		& $3$
		& \cite{Luo1} \\
		\cmidrule{2-5}
	
		& $\begin{aligned}
			d&=(h+1)(2^k-1)+1,\\
			h&(2^t-1)\equiv 1\pmod{2^k+1}
		\end{aligned}$
		& $\begin{aligned}
			n&=2k,\ k,t\text{ odd},\\
			1&\le t<k,\ \gcd(k,t)=1;\\
			l&=0,1,2
		\end{aligned}$
		& $3$
		& \cite{Fan1} \\
		\midrule
		prime
		& $d=2p^k-1$, $p\equiv2\pmod 3$
		& $\begin{aligned}
			n&=2k,\ k \text{ odd},\\
			l&=0,1,2
		\end{aligned}$
		& $3$
		& \cite{Xia2} \\
		\midrule
		odd prime
		& $\displaystyle d=\left(\frac{p^{2k}+1}{2}\right)^2$
		& $\begin{aligned}
			n&=4k, \\
			0&\le l<\frac{p^{2k}+1}{2}
		\end{aligned}$
		
		& $\displaystyle\frac{p^{2k}+1}{2}$
		& \cite{Seo2} \\
		\midrule
		$2$
		& $\displaystyle d = \frac{2^s}{2^s+1}(2^k-1)+1$
		& $\begin{aligned}
			n&=2k,\ 1<s<2k, v_2(s)>v_2(k)\\
			0&\le l\le 2^r,r=\gcd(s,k);
		\end{aligned}$
		& $2^r+1$
		& This paper \\
		\bottomrule
	\end{tabular}
\end{table}

\section{Preliminaries}\label{pre}

Let $n$ and $k$ be two positive integers with $n=2k$, $\mathbb{F}_{2^n}$ denote a finite field with $2^n$ elements,  and $\mathbb{F}_{2^n}^{*}$ be the multiplicative group of $\mathbb{F}_{2^n}$.
For a positive integer $e$ with $e\,|\,n$, the {\it trace function} ${\rm Tr}_{e}^{n}(\cdot)$ from $\mathbb{F}_{2^n}$ to its subfield $\mathbb{F}_{2^e}$ is defined by
\[
{\rm Tr}_{e}^{n}(x)=\sum_{i=0}^{\frac{n}{e}-1} x^{2^{ei}}.
\]

For each element $x$ in the finite field $\mathbb{F}_{2^n}$, let $\bar{x}=x^{2^{k}}$. The {\it unit circle}  of $\mathbb{F}_{2^n}$ is the set
\begin{equation}\label{2.1}
	S=\left\{z \in \mathbb{F}_{2^n}: z\overline{z}=1\right\}.
\end{equation}
Since $\gcd(2^k-1,2^k+1)=1$, every element of
$\mathbb F_{2^n}^{*}$ has a unique representation $x=zy$ \cite{Niho1,Rose1}, where  $z \in S$ and $y \in \mathbb{F}_{2^k}^*$.

Let $\{u(t)\}$ be a binary $m$-sequence of period $2^n-1$, and we can assume that it is given by
\begin{equation}\label{2.2}
	u(t)=\mathrm{Tr}_{1}^{n}\left(\alpha^{t}\right),\ 0 \leq t \leq 2^n-2,
\end{equation}
where $\alpha$ is a primitive element of $\mathbb{F}_{2^n}$.
For a decimation $d$, let $\{v_l(t)=u(dt+l)\}$ be the {\it $l$-th $d$-decimation sequence} of $\{u(t)\}$, where $l$ is an integer with $0\leq l < \gcd(d,2^n-1)$. Let $L=\frac{2^n-1}{\gcd(d,2^n-1)}$.
Then $\{v_l(t)\}$ is a binary sequence whose least period
divides $L$. If $L>1$ and the multiplicative order of $2$
modulo $L$ is $n$, then its least period is exactly $L$. When $\gcd(d,2^n-1)=1$, there is only one $d$-decimated
sequence up to cyclic shifts, namely $\{v_0(t)\}$.
We denote it by $\{v(t)\}$ and refer to it as the
\emph{$d$-decimated sequence} of $\{u(t)\}$
\cite{Niho1,Helle2}.

The {\it cross-correlation function} \cite{Niho1,Helle2} between $\{u(t)\}$ and $\{v_l(t)\}$ at shift $\tau$ is defined by
\begin{align}\label{2.3}
	C_{d, l}(\tau)&=\sum\limits_{t=0}^{2^n-2} (-1)^{u(t+\tau)+v_l(t)}\nonumber \\
&=\sum\limits_{t=0}^{2^n-2} (-1)^{u(t+\tau)+u(dt+l)}\\
&=\sum\limits_{x \in \mathbb{F}_{2^n}^{*}} (-1)^{\mathrm{Tr}_{1}^{n}\left(ax^d+bx\right)},\nonumber
\end{align}
where $0 \leq \tau < 2^n-1$, $a=\alpha^l$ and $b=\alpha^{\tau}$. For each fixed $l$, the element $a$ is fixed, while $b$
runs through $\mathbb F_{2^n}^{*}$ as $\tau$ runs through
a complete period.

The cross-correlation values in \eqref{2.3} are closely
related to the numbers of solutions of sparse equations
over finite fields. The following root-count result will
be used to determine the possible cross-correlation values.
It follows from Lemma 22 and the unit-circle argument
in the proof of Theorem 23 in \cite{Dobb2}.

\begin{lem}{\rm(\!\!{\cite{Dobb2}})}\label{lem 1}
	Let $n$, $k$ and $s$ be positive integers with $n=2k$. Then the equation
	\[
	x^{2^s+1} + a x^{2^s} + b x + c=0,\quad a, b, c \in \mathbb{F}_{2^n}
	\]
	has either $0$, $1$, $2$ or $2^{\gcd(s,k)}+1$ distinct roots in the set $S$, which is given in (\ref{2.1}).
\end{lem}

The following conclusion concerns the values of a class of exponential sums over finite fields.
	
	\begin{lem} {\rm(\!\!\cite{Helle1,Fan1})}\label{lem 2}
		Let $n=2k$, and $N$ be an integer with $N>2$, $N\,|\,(2^n-1)$ and $N\,|\,(2^i+1)$ for some $i$. Suppose $j$ is the least positive integer such that $N\,|\,(2^j+1)$. Then
		\[
		\sum_{x \in \mathbb{F}_{2^n}}(-1)^{\mathrm{Tr}_{1}^{n}\left(cx^N\right)}=
		\begin{cases}
			2^n, & \text{if } c=0, \\
			(-1)^{\frac{k}{j}+1}(N-1)2^{k},& \text{if } c \text{ is an } N\text{-th power in } \mathbb{F}_{2^n}^{*}, \\
			(-1)^{\frac{k}{j}}2^{k}, & \text{if } c \text{ is not an } N\text{-th power in } \mathbb{F}_{2^n}^{*}.
		\end{cases}
		\]
\end{lem}

We recall the following standard identities for the greatest common divisors \cite{Kla1,Finite}.

\begin{lem}\label{lem 3}
	Let $t$ and $l$ be two positive integers. Then $\gcd(2^t-1,2^l-1)=2^{\gcd(t,l)}-1$,
	\[
	\gcd\left(2^t+1, 2^l-1\right)=
	\begin{cases}
		1, & \text{if } \frac{l}{\gcd(t,l)} \text{ is odd}, \\
		2^{\gcd(t,l)}+1, & \text{if } \frac{l}{\gcd(t,l)} \text{ is even}
	\end{cases}
	\]
	and
	\[
	\gcd\left(2^t+1, 2^l+1\right)=
	\begin{cases}
		1, & \text{if one of } \frac{t}{\gcd(t,l)},\frac{l}{\gcd(t,l)} \text{ is even}, \\
		2^{\gcd(t,l)}+1, & \text{if both } \frac{t}{\gcd(t,l)} \text{ and } \frac{l}{\gcd(t,l)} \text{ are odd}.
	\end{cases}
	\]
	
\end{lem}

The following result is a consequence of Lemma \ref{lem 3}.

\begin{lem}\label{lem3}
	Let \(n=2k\) and $d$ be the Niho-type decimation defined  in \eqref{new-d}.  Set $r=\gcd(s,k)$. Then
	\[
	\gcd(d,2^n-1)=2^r+1
	\]
\end{lem}

\begin{proof}
	 Since $v_2(s)>v_2(k)$, $\frac{s}{r}$ is even and
$\frac{k}{r}$ is odd. Lemma \ref{lem 3} therefore gives
	\begin{equation}\label{2.4}
		\gcd(2^s+1,2^k+1)=1.\end{equation}
	Thus, $2^s+1$ is invertible modulo $2^k+1$, and $d$ is well defined.
	
	Since $d\equiv1\pmod{2^k-1}$, we have
$\gcd(d,2^k-1)=1$. Moreover, since
$2^n-1=(2^k-1)(2^k+1)$ and
$\gcd(2^k-1,2^k+1)=1$, it follows that $$\gcd(d,2^n-1)=\gcd(d,2^k+1).$$
	
	Note that
	\[
	(2^s+1)d
	\equiv
	2^s(2^k-1)+(2^s+1)
	\equiv
	1-2^s
	\pmod{2^k+1}.
	\]
	Then equality (\ref{2.4}) implies $$\gcd(d,2^k+1)=\gcd((2^s+1)d,2^k+1)=\gcd(2^s-1,2^k+1).$$
Since $\frac{s}{r}$ is even, by Lemma \ref{lem 3}, we get $$\gcd(2^s-1,2^k+1)=2^{\gcd(s,k)}+1=2^r+1.$$
	
Combining the above equalities,  the desired result is obtained.
\end{proof}

The following two results will be used in Section \ref{sec3}
to evaluate $B_{2^r+1}(a,\beta)$, which is defined in \eqref{B}.

\begin{lem}\label{tr}
	 Let $n$, $k$, $s$, and $r$ be as in \eqref{new-d}
and Lemma~\ref{lem3}. Define
		\begin{equation}\label{V}V=\left\{v\in\F_{2^n}:\Tr_k^n(v)=1\right\}.\end{equation}
	Then,
		\[
		\left\{
		v^{2^s}+v:v\in V
		\right\}
		=
		\left\{
		y\in\F_{2^k}:\Tr_r^k(y)=0
		\right\}=\left\{x^{2^s}+x: x \in \F_{2^k} \right\}.
		\]
		This set has cardinality $2^{k-r}$. Moreover, for every $y\in\F_{2^k}$ with $\Tr_r^k(y)=0$, we have
$\left|\left\{v\in V:v^{2^s}+v=y\right\}
	\right|
	=2^r.$
	\end{lem}

	\begin{proof}
		For every $v\in V$, it can be verified that $\left(v^{2^s}+v\right)^{2^k}
			=(v^{2^k})^{2^s}+v^{2^k}=(v+1)^{2^s}+v+1=v^{2^s}+v$ due to $\Tr_k^n(v)=1$, i.e., $v^{2^k}=v+1$.
		Thus $v^{2^s}+v\in\F_{2^k}$.

For the complete residue system $\left\{0,1,\cdots,\frac{k}{r}-1\right\}$ modulo $\frac{k}{r}$, by $\gcd\left(\frac{s}{r},\frac{k}{r}\right)=1$, we have
that $\left\{0,\frac{s}{r},2\cdot\frac{s}{r}\cdots,(\frac{k}{r}-1)\cdot\frac{s}{r}\right\}$ is also a complete residue system modulo $\frac{k}{r}$
\cite[Theorem~2.6, p.~51]{Niven}. Multiply each residue class in these two complete residue systems by $r$, then we have
\[
\left\{is\bmod k:0\leq i<\frac{k}{r}\right\}=\left\{ir\bmod k:0\leq i<\frac{k}{r}
\right\}.
\]
This shows that for $y\in \F_{2^k}$,
$\Tr_r^k\left(y\right)
=\sum\limits_{i=0}^{\frac{k}{r}-1}y^{2^{ir}}
=\sum\limits_{i=0}^{\frac{k}{r}-1}y^{2^{is}}	
$
and then by $v^{2^s}+v\in\F_{2^k}$,
$$
	\Tr_r^k\left(v^{2^s}+v\right)
	=\sum\limits_{i=0}^{\frac{k}{r}-1}\left(v^{2^s}+v\right)^{2^{is}}
	=\sum\limits_{i=0}^{\frac{k}{r}-1}\left(v^{2^{(i+1)s}}+v^{2^{is}}\right)	
	=\sum\limits_{j=1}^{\frac{k}{r}}v^{2^{js}}+\sum_{i=0}^{\frac{k}{r}-1}v^{2^{is}}	
	=v^{2^{\frac{ks}{r}}}+v.
$$
The fact $v_2(s)>v_2(k)$ implies that $\frac{s}{r}$ is even. This together with $v^{2^k}=v+1$ gives
$$v^{2^{\frac{ks}{r}}}=(v+1)^{2^{k\left(\frac{s}{r}-1\right)}}=v^{2^{k\left(\frac{s}{r}-1\right)}}+1=(v+1)^{2^{k\left(\frac{s}{r}-2\right)}}+1=v^{2^{k\left(\frac{s}{r}-2\right)}}=\cdots=v.$$
Thus, $\Tr_r^k\left(v^{2^s}+v\right)=0$
and therefore
\begin{equation}\label{<}
		\left\{v^{2^s}+v:v\in V\right\}\subseteq\left\{	y\in\F_{2^k}:\Tr_r^k(y)=0
		\right\}.
\end{equation}

For a given $v_0\in V$, let $v\in V$ satisfy $v^{2^s}+v=v_0^{2^s}+v_0$. Then, $(v+v_0)^{2^s}=v+v_0$.
Since $v,v_0\in V$, we also have $(v+v_0)^{2^k}=v^{2^k}+v_0^{2^k}=v+v_0$. Thus, $v+v_0\in\F_{2^k}$ and
$(v+v_0)^{2^s}=v+v_0$. Since $r=\gcd(s,k)$,
it follows that $v+v_0\in\F_{2^r}$. Conversely, if $\mu\in\F_{2^r}$, then $r\,|\,s$ and $r\,|\, k$ imply $\mu^{2^s}=\mu$ and $\mu^{2^k}=\mu$, respectively. Thus, we have
\[
(v_0+\mu)^{2^s}+(v_0+\mu)=v_0^{2^s}+v_0+\mu+\mu=v_0^{2^s}+v_0
\]
and\[\Tr_k^n(v_0+\mu)=\Tr_k^n(v_0)+\mu\Tr_k^n(1)=1.
\]
So $v_0+\mu\in V$.

Thus, for each fixed $v_0\in V$, the elements $v\in V$
satisfying $v^{2^s}+v=v_0^{2^s}+v_0$ are precisely
$v=v_0+\mu$ with $\mu\in\F_{2^r}$.
Hence each $y\in\{v^{2^s}+v:v\in V\}$ has exactly
$2^r$ preimages in $V$. Since $|V|=2^k$, we obtain
\begin{equation}\label{car}
	\left|\left\{v^{2^s}+v: v \in V \right\}\right|=\frac{2^k}{2^r}=2^{k-r}.
\end{equation}
On the other hand, since $\Tr_r^k(\cdot)$ is a surjection from $\F_{2^k}$ to $\F_{2^r}$,
\begin{equation}\label{yk}
\left|\left\{y\in\F_{2^k}:\Tr_r^k(y)=0\right\}\right|=2^{k-r}.
\end{equation}
Combining \eqref{<}, \eqref{car} and \eqref{yk} yields
\begin{equation}\label{seteq}
\left\{v^{2^s}+v: v \in V \right\}	=\left\{y\in\F_{2^k}:\Tr_r^k(y)=0\right\}.
\end{equation}
Thus, for every $y\in\F_{2^k}$ with $\Tr_r^k(y)=0$, the preceding preimage count gives
$$\left|\left\{v\in V:v^{2^s}+v=y\right\}
	\right|
	=2^r.$$

Moreover, according to \eqref{seteq}, we can choose an element $v_0\in V$ such that $v_0^{2^s}+v_0=0$. Then, $V$ defined in \eqref{V} can be expressed as
$V=\{v_0+x: x\in \F_{2^k}\}$. Thus,
$$\left\{v^{2^s}+v: v \in V \right\}=\left\{(v_0+x)^{2^s}+(v_0+x): x \in \F_{2^k} \right\}=\left\{x^{2^s}+x: x \in \F_{2^k} \right\}.$$
This completes the proof.
\end{proof}

\begin{cor}\label{y}
	 With the notation introduced in Lemma \ref{tr}, for a given element $\delta\in \F_{2^k}$ satisfying $\Tr_r^k(\delta)=0$, there are exactly $2^{k-r}-1$ elements $c\in \F_{2^k}^*$ such that the equation
		\begin{equation}\label{x}
			x^{2^s}+x=\delta c^{2^s}+c
		\end{equation}
in variable $x$ has a solution in $\F_{2^k}$.
	\end{cor}

	\begin{proof} By Lemma \ref{tr}, the equation \eqref{x} has solutions if and only if  $\Tr_r^k\left(\delta c^{2^s}+c\right)=0$.
Since $\delta^{2^{2k}}=\delta$ and $r\mid s$,
		\[\Tr_r^k\left(\delta c^{2^s}+c\right)=\Tr_r^k\left(\left(\delta^{2^{2k-s}}c\right)^{2^s}+c\right)=\Tr_r^k\left(\delta^{2^{2k-s}} c+c\right)=\Tr_r^k\left( \left(\delta^{2^{2k-s}}+1\right)c\right).\]
Moreover, since $r\mid(2k-s)$ and $\frac{k}{r}$ is odd, we have $$\mathrm{Tr}_r^k\left(\delta^{2^{2k-s}}+1\right)=\mathrm{Tr}_r^k(\delta)+\mathrm{Tr}_r^k(1)=0+1=1.$$  Thus, $\delta^{2^{2k-s}}+1\ne 0$, and the map
$c\mapsto\left(\delta^{2^{2k-s}}+1\right)c$ is a bijection
on $\F_{2^k}^*$. By \eqref{yk}, exactly $2^{k-r}-1$ nonzero elements $y$ of $\F_{2^k}$
satisfy $\Tr_r^k(y)=0$.  Consequently, there are exactly $2^{k-r}-1$ elements $c\in \F_{2^k}^*$ such that the equation  has a solution in $\F_{2^k}$. This completes the proof.
		 \end{proof}

\section{Four-valued cross-correlation functions}\label{sec3}

In this section, we determine the cross-correlation
distribution between the binary $m$-sequence $\{u(t)\}$
and each of its short-period $d$-decimated sequences,
where $d$ is defined in \eqref{new-d}.

Throughout this section, the parameters $n$, $k$, and $s$
are as in \eqref{new-d}, and we write $r=\gcd(s,k)$. The condition $v_2(s)>v_2(k)$ implies that $\frac{s}{r}$ is even
and $\frac{k}{r}$ is odd. By Lemma \ref{lem 3},
\begin{equation}\label{3.1}
\begin{aligned}
\gcd(2^s-1,2^k+1) &= 2^r+1,\\
\gcd(2^s+1,2^k+1) &= 1,\\
\gcd(2^s-1,2^k-1) &= 2^r-1.
\end{aligned}
\end{equation}
Moreover, Lemma \ref{lem3} gives
$\gcd(d,2^n-1)=2^r+1$. Thus, there are $2^r+1$ short-period $d$-decimated sequences $\{v_l(t)\}$, indexed by $0\le l\le2^r$.
We first verify that each sequence $\{v_l(t)\}$ has the least period
$L=\frac{2^n-1}{2^r+1}$. Since $s\ne k$ and $1<s<2k$,
we have $r<k$, and hence
\[
L=(2^k-1)\frac{2^k+1}{2^r+1}>2^k-1.
\]
The multiplicative order of $2$ modulo $L$ divides
$n=2k$. Any proper divisor of $n$ is at most $k$,
so the above inequality forces this order to be $n$.
The least-period criterion stated in Section~\ref{pre} then gives
the claim.

For $a, b \in \F_{2^{n}}$, define
\begin{equation}
	W_d(a,b)=\sum_{x\in\F_{2^{n}}}(-1)^{\Tr_1^{n}(ax^{d}+bx)}.
	\label{3.2}
\end{equation}
By (\ref{2.3}), we have
\begin{equation}\label{CC}
	C_{d,l}(\tau)=W_d(a,b)-1,
\end{equation}
where $a=\alpha^l$ and $b=\alpha^\tau$.  Thus,  for each fixed $l\in\{0,1,\ldots,2^r\}$,
to determine the cross-correlation values and their
frequencies between $\{u(t)\}$ and its $l$-th
$d$-decimated sequence $\{v_l(t)\}$, it suffices
to determine the value distribution of
$W_d(\alpha^l,\alpha^\tau)$ as $\tau$ ranges
from $0$ to $2^n-2$.

\begin{prop}\label{lem 4}
	Let $d$ be given in \eqref{new-d} and $W_d(a,b)$ be defined in \eqref{3.2}. For every $a,b\in\F_{2^n}^{*}$, we have  $W_d(a,b)\in\{-2^k,0,2^k,2^{k+r}\}$.
\end{prop}

\begin{proof} Note that every nonzero element $x\in\mathbb{F}_{2^n}$ can be uniquely represented as $x=zy$ for $z\in S$ and $y\in \mathbb{F}_{2^k}^*$. Then for each pair $a,b\in \mathbb{F}_{2^n}^*$, we have
\begin{align}\label{3.3}
	W_d(a,b)&=\sum_{x\in\F_{2^{n}}}(-1)^{\Tr_1^{n}(ax^{d}+bx)}\nonumber \\
&=1+\sum\limits_{z\in S}\sum\limits_{y\in\mathbb{F}_{2^k}^*}(-1)^{{\rm Tr}^n_1\left(a(zy)^d+bzy\right)}\nonumber\\
&=-2^k+\sum\limits_{z\in S}\sum\limits_{y\in\mathbb{F}_{2^k}}(-1)^{{\rm Tr}^n_1\left(az^dy+bzy\right)}\\
&=-2^k+\sum\limits_{z\in S}\sum\limits_{y\in\mathbb{F}_{2^k}}(-1)^{{\rm Tr}^k_1\left(\left(az^d+bz+(az^d+bz)^{2^k}\right)y\right)}\nonumber\\
&=-2^k+\left|\left\{z\in S: az^d+bz+(az^d+bz)^{2^k}=0\right\}\right|\cdot 2^k\nonumber
\end{align}

For $z\in S$, consider the equation
$$az^d+bz+(az^d+bz)^{2^k}=0.$$
Using $\bar z=z^{-1}$ and the expression for $d$
in \eqref{new-d}, we can rewrite this equation as
$$az^{1-2\cdot\frac{2^s}{2^s+1}}+bz+\bar{a}z^{2\cdot\frac{2^s}{2^s+1}-1}+\bar{b}z^{-1}=0.$$
Multiplying both sides by $z$ and then renaming
$z^2$ as $z$, we obtain the equivalent equation
\begin{equation}\label{3.6}
az^{1-\frac{2^s}{2^s+1}}+bz+\bar{a}z^{\frac{2^s}{2^s+1}}+\bar{b}=0.
\end{equation}
Note that $\gcd(2^s+1,2^k+1)=1$.  Replacing $z$ with $z^{2^s+1}$ in (\ref{3.6}), we have
 \begin{equation}\label{3.61}
 bz^{2^s+1}+\bar{a}z^{2^s}+az+\bar{b}=0,
 \end{equation}
which has either $0$, $1$, $2$, or $2^{\gcd(s,k)}+1=2^r+1$ solutions in $S$
by Lemma~\ref{lem 1}.
Therefore, from (\ref{3.3}), $W_d(a,b)$ equals either $-2^k$, $0$, $2^k$, or $2^{k+r}$.
\end{proof}

Proposition \ref{lem 4} gives the possible values of $W_d(a,b)$. To determine the frequencies of these values, we define
\begin{equation}
	\begin{aligned}	
		N(a,b)
		:=&\left|\left\{z\in S: bz^{2^s+1}+\bar{a}z^{2^s}+az+\bar{b}=0\right\}\right|.
	\end{aligned}	\label{Nab}
\end{equation}
 Then, $W_d(a,b)=2^k\bigl(N(a,b)-1\bigr)$. For each given $a=\alpha^l$, determining the frequencies of the values of $W_d(a,\alpha^{\tau})$ as $\tau$ runs through each integer from $0$ to $2^n-2$ is equivalent to counting, for each
$i\in\{0,1,2,2^r+1\}$, the elements
$b\in\F_{2^n}^*$ satisfying $N(a,b)=i$.  For $a=\alpha^l$, define
\begin{equation}\label{Ai}
	A_i(a):=\left|\left\{b\in\mathbb F_{2^n}^{*}:N(a,b)=i\right\}\right|.
	\end{equation}
 Therefore, to determine the correlation distribution, it suffices to determine $A_0(a)$, $A_1(a)$, $A_2(a)$, and $A_{2^r+1}(a)$. In the following, we first determine the value of $A_{2^r+1}(a)$ directly.

\begin{prop}\label{pro1} Let $a\in\F_{2^n}^*$ and $b\in\F_{2^n}$, and suppose that $\beta\in S$ is a
solution of \eqref{3.61}. Then there is a unique
$\gamma\in\F_{2^k}$ such that
\[
b=\bar{a}\beta^{-1}
  +\gamma\beta^{2^{k-1}(2^s+1)}.
\]
With this $\gamma$, the map
\[
z\longmapsto \frac{z}{z+\beta}
\]
gives a bijection between the roots of \eqref{3.61}
in $S\setminus\{\beta\}$ and the solutions in $V$
of \begin{equation}\label{affine}\left[\beta^{2^{k-1}(2^s+1)}\left(\bar a\beta^{2^s}+a\beta\right)+\gamma\right]v^{2^s}	+\gamma v+a\beta^{2^{k-1}(2^s+1)+1}+\gamma=0\end{equation}
in variable $v$, where $V$ is
defined in Lemma \ref{tr}. Consequently, for every positive integer $i\in\{1,2,2^r+1\}$,
the equation \eqref{3.61} has exactly $i$ roots in $S$
if and only if \eqref{affine} has exactly
$i-1$ solutions in $V$.
\end{prop}

\begin{proof} Since $\beta\in S$ is a solution of \eqref{3.61}, we have $b\beta^{2^s+1}+\bar b=\bar a\beta^{2^s}+a\beta$. Consider the equation
\begin{equation}\label{3.62}\bar{x}+\beta^{2^s+1}x=\bar a\beta^{2^s}+a\beta
\end{equation}in variable $x$. The element $x=\bar a\beta^{-1}$ is a solution of \eqref{3.62}. It can be verified that for every $\gamma\in\F_{2^k}$, the element $x=\gamma\beta^{2^{k-1}(2^s+1)}$
	satisfies the homogeneous equation $\bar{x}+\beta^{2^s+1}x=0$ due to
	\[
	\begin{aligned}
		&\overline{\gamma\beta^{2^{k-1}(2^s+1)}}+\beta^{2^s+1}\cdot\gamma\beta^{2^{k-1}(2^s+1)}\\
		&=\gamma\beta^{-2^{k-1}(2^s+1)}+\gamma\beta^{(2^s+1)(1+2^{k-1})}\\
		&=\gamma\beta^{-2^{k-1}(2^s+1)}+\gamma\beta^{(2^s+1)(-2^{k-1}+2^k+1)}\\
		&=\gamma\beta^{-2^{k-1}(2^s+1)}\left(1+\beta^{(2^s+1)(2^k+1)}\right)\\
		&=\gamma\beta^{-2^{k-1}(2^s+1)}(1+1)\\
		&=0.
	\end{aligned}
	\]
Since $\beta\ne0$, the equation $\bar{x}+\beta^{2^s+1}x=0$ has
$2^k$ distinct solutions.
On the other hand, the equation \eqref{3.62} has at most
$2^k$ solutions in $\F_{2^n}$.
Therefore, every solution of \eqref{3.62} can be
written uniquely as
\[
x=\bar a\beta^{-1}
+\gamma\beta^{2^{k-1}(2^s+1)},
\qquad \gamma\in\F_{2^k}.
\]
In particular, for every $b\in\F_{2^n}$, there is a unique $\gamma\in\F_{2^k}$
such that
\[
b=\bar a\beta^{-1}
+\gamma\beta^{2^{k-1}(2^s+1)}.
\]

Let $z$ be an element in the set $S\setminus\{\beta\}$, and define a mapping
$$
\begin{array}{rcl}
\Psi: S\setminus\{\beta\}&\longrightarrow &V\\
z\,\,\,\,\,\,\,\,\,&\mapsto &\frac{z}{z+\beta}.
\end{array}
$$
Note that $z, \beta\in S$, and then we have $z^{2^k}=z^{-1}$ and $\beta^{2^k}=\beta^{-1}$. This gives
	$$(\Psi(z))^{2^k}=\frac{z^{-1}}{z^{-1}+\beta^{-1}}=\frac{\beta}{z+\beta}=\Psi(z)+1,$$
i.e., $(\Psi(z))^{2^k}+\Psi(z)=1$. Therefore, $\Psi$ is indeed a mapping from $S\setminus\{\beta\}$ to $V$. Further, if $z_1,z_2$ are two elements in $S\setminus\{\beta\}$
satisfy $\Psi(z_1)=\Psi(z_2)$, that is to say,
$$\frac{z_1}{z_1+\beta}=\frac{z_2}{z_2+\beta},$$
then we have $z_1=z_2$. This shows that $\Psi$ is injective. For every element $v\in V$, we have $v\ne 0,1$ because of $v^{2^k}=v+1$, and it can be verified that
	\[
	\left(\frac{\beta v}{v+1}\right)^{2^k}=\frac{\beta^{-1} v^{2^k}}{v^{2^k}+1}=\frac{ v+1}{\beta v}=\left(\frac{\beta v}{v+1}\right)^{-1},
	\]
i.e., $\frac{\beta v}{v+1}\in S\setminus\{\beta\}$, and
\[
	\Psi\left(\frac{\beta v}{v+1}\right)=\frac{\frac{\beta v}{v+1}}{\frac{\beta v}{v+1}+\beta}=v.
	\]
This shows that $\Psi$ is surjective. Thus, $\Psi$ is a bijection from $S\setminus\{\beta\}$ to $V$.

 Under this bijection, the solutions of \eqref{3.61}
in $S\setminus\{\beta\}$ correspond precisely to
the solutions in $V$ of the linearized equation
derived below.	
	
Let $z$ be a solution of \eqref{3.61} different from  $\beta$. Then from $bz^{2^s+1}+\bar{a}z^{2^s}+az+\bar{b}=0$ and $b\beta^{2^s+1}+\bar{a}\beta^{2^s}+a\beta+\bar{b}=0$, we have
\begin{equation}\label{3.63}
b\left(z^{2^s+1}+\beta^{2^s+1}\right)+\bar a\left(z^{2^s}+\beta^{2^s}\right)+a(z+\beta)=0.
\end{equation}
For every solution $z$ of \eqref{3.61} different from $\beta$, there is a unique element $v\in V$ such that $\Psi(z)=v$ and then $z=\frac{\beta v}{v+1}$.
Substituting $z=\frac{\beta v}{v+1}$ into \eqref{3.63}, we obtain
	\[
\begin{aligned}
	&\,b\left(\left(\frac{\beta v}{v+1}\right)^{2^s+1}+\beta^{2^s+1}\right)+\bar a\left(
	\left(\frac{\beta v}{v+1}\right)^{2^s}+\beta^{2^s}\right)+a\left(\left(\frac{\beta v}{v+1}\right)+\beta\right)=0\\
\Leftrightarrow&\,b\left(\beta^{2^s+1}\left(v^{2^s}+v+1\right)\right)+\bar a\left(
	\beta^{2^s}(v+1)\right)+a\left(\beta\left(v^{2^s}+1\right)\right)=0\\
\Leftrightarrow&\,\left(\bar a\beta^{-1}+\gamma\beta^{2^{k-1}(2^s+1)}\right)\left(\beta^{2^s+1}\left(v^{2^s}+v+1\right)\right)+\bar a\left(
	\beta^{2^s}(v+1)\right)+a\left(\beta\left(v^{2^s}+1\right)\right)=0\\
\Leftrightarrow&\,\bar a\beta^{2^s}v^{2^s}+\gamma\beta^{(2^{k-1}+1)(2^s+1)}v^{2^s}+\gamma\beta^{(2^{k-1}+1)(2^s+1)}v+\gamma\beta^{(2^{k-1}+1)(2^s+1)}+a\beta v^{2^s}+a\beta=0\\
\Leftrightarrow&\,\left(\bar a\beta^{2^s}+\gamma\beta^{(2^{k-1}+1)(2^s+1)}+a\beta\right) v^{2^s}
+\gamma\beta^{(2^{k-1}+1)(2^s+1)}v+\gamma\beta^{(2^{k-1}+1)(2^s+1)}+a\beta=0\\
\Leftrightarrow&\left[\beta^{2^{k-1}(2^s+1)}\left(\bar a\beta^{2^s}+a\beta\right)+\gamma\right]v^{2^s}	+\gamma v+a\beta^{2^{k-1}(2^s+1)+1}+\gamma=0,
	\end{aligned}
\]
where the third equation holds since $b=\bar a\beta^{-1}+\gamma\beta^{2^{k-1}(2^s+1)}$.  The last equivalence follows by multiplying by
$\beta^{2^{k-1}(2^s+1)}$ and using $\beta^{2^k+1}=1$.

The bijection $\Psi$ and the above calculations show
that when $N(a,b)\ge 1$, i.e., $N(a,b)\in \{1, 2, 2^r+1\}$,  one has
$$N(a,b)=1+\left|\{v\in V:v\text{ satisfies (3.9)}\}\right|,$$   which implies the desired result.

\end{proof}

Fix $a\in\F_{2^n}^*$.
To determine $A_{2^r+1}(a)$, we need to count
the elements $b\in\F_{2^n}^*$ such that
the equation \eqref{3.61} has exactly $2^r+1$
solutions in $S$.
For each $\beta\in S$, define
\begin{equation}\label{B}
B_{2^r+1}(a,\beta):=
\left|\left\{
b\in\F_{2^n}^*:
N(a,b)=2^r+1
\text{ and $\beta$ is a solution of \eqref{3.61}}
\right\}\right|.
\end{equation}
Thus, for fixed $a$ and $\beta$, $B_{2^r+1}(a,\beta)$
is the number of nonzero elements $b$ such that
\eqref{3.61} has exactly $2^r+1$
solutions in $S$, including $\beta$.

To relate $A_{2^r+1}(a)$ to $B_{2^r+1}(a,\beta)$, we count the ordered pairs
$(b,\beta)\in\F_{2^n}^*\times S$ such that
$N(a,b)=2^r+1$ and $\beta$ is a solution of
\eqref{3.61}.
For each fixed $b$ counted by $A_{2^r+1}(a)$,
there are exactly $2^r+1$ choices of $\beta$.
On the other hand, for each fixed $\beta\in S$,
there are exactly $B_{2^r+1}(a,\beta)$ choices of $b$.
Counting these pairs in two ways gives
\begin{equation}\label{AB}
(2^r+1)A_{2^r+1}(a)
=\sum_{\beta\in S}B_{2^r+1}(a,\beta).
\end{equation} It therefore suffices to evaluate $B_{2^r+1}(a,\beta)$
for each $\beta\in S$.

To evaluate $B_{2^r+1}(a,\beta)$, fix $\beta\in S$.
For fixed $a$ and $\beta$, Proposition \ref{pro1}
shows that
\begin{equation}\label{bb}
b=\bar a\beta^{-1}
+\gamma\beta^{2^{k-1}(2^s+1)},
\quad \gamma\in\F_{2^k},
\end{equation}
gives a bijection between the parameters
$\gamma\in\F_{2^k}$ and the elements
$b\in\F_{2^n}$ such that $\beta$ is a solution
of \eqref{3.61}.
For each corresponding pair $(b,\gamma)$,
Proposition \ref{pro1} further shows that \eqref{3.61} has exactly $2^r+1$
solutions in $S$ if and only if
\eqref{affine} has exactly $2^r$ solutions in $V$.
Therefore, determining $B_{2^r+1}(a,\beta)$ reduces to
counting the parameters $\gamma\in\F_{2^k}$
such that the corresponding $b$ is nonzero
and the  equation \eqref{affine} has exactly $2^r$
solutions in $V$.

\begin{prop}\label{lem5}
	For $a\in\F_{2^n}^*$ and $\beta\in S$, we have
	\[
	B_{2^r+1}(a,\beta)=\begin{cases}\displaystyle
		\frac{2^k-2^r}{2^r},
		& if \,\,\,
		\bar a\beta^{2^s}+a\beta=0,
		\\[10pt]
		\displaystyle
		\frac{2^k-2^r}{2^r(2^r-1)},
		&
		if\,\,\, \bar a\beta^{2^s}+a\beta\ne0.
	\end{cases}
	\]
\end{prop}	
\begin{proof}	
  The coefficient of $v^{2^s}$ in \eqref{affine} is $\beta^{2^{k-1}(2^s+1)}	\left(\bar a\beta^{2^s}+a\beta	\right)+\gamma$. We first show that this coefficient belongs to $\F_{2^k}$.
Since $\beta\in S$, we have $\bar\beta=\beta^{-1}$.
Moreover,
\[
-2^{k-1}(2^s+1)-1
\equiv 2^{k-1}(2^s+1)+2^s
\pmod{2^k+1}.
\]
Hence
\begin{equation}\label{xaddeq1}
\begin{aligned}
\beta^{2^{k-1}(2^s+1)}
\left(\bar a\beta^{2^s}+a\beta\right)
&=\overline{a\beta^{2^{k-1}(2^s+1)+1}}
  +a\beta^{2^{k-1}(2^s+1)+1}\\
&=\Tr_k^n\left(a\beta^{2^{k-1}(2^s+1)+1}\right)
\in\F_{2^k}.
\end{aligned}
\end{equation}
Since $\gamma\in\F_{2^k}$, the coefficient of
$v^{2^s}$ also belongs to $\F_{2^k}$.
We now evaluate $B_{2^r+1}(a,\beta)$ in the following
two cases.
	
\noindent\emph{Case 1:
		$\bar a\beta^{2^s}+a\beta=0$.}

	In this case, the equation \eqref{affine} becomes $\gamma v^{2^s}+\gamma v+a\beta^{2^{k-1}(2^s+1)+1}+\gamma=0$ and it can be written as $\gamma\left(v^{2^s}+v\right)+a\beta^{2^{k-1}(2^s+1)+1}+\gamma=0$.
	Since $a\ne0$ and $\beta\ne0$, we have $a\beta^{2^{k-1}(2^s+1)+1}\ne 0$ and then $\gamma\ne0$. Thus, we have
	\begin{equation}\label{vr}
		v^{2^s}+v=1+\frac{a\beta^{2^{k-1}(2^s+1)+1}}{\gamma}.
	\end{equation}
Moreover, in this case, the preceding identity \eqref{xaddeq1} gives
\[
\Tr_k^n\left(a\beta^{2^{k-1}(2^s+1)+1}\right)=0.
\]
Thus,
$a\beta^{2^{k-1}(2^s+1)+1}\in\F_{2^k}^*$,
since $a\ne0$ and $\beta\ne0$.

By Lemma \ref{tr}, for given $a, \beta, \gamma$, if the equation \eqref{vr} in variable $v$ is solvable, then it has exactly $2^r$ solutions in $V$.
Now count the elements
$\gamma\in\F_{2^k}^*$ such that \eqref{vr}
has a solution in $V$.

Define a mapping $\Upsilon$ from $\F_{2^k}^{*}$
to $\F_{2^k}\setminus\{1\}$ by
$$
\begin{array}{rcl}
\Upsilon: \F_{2^k}^{*}&\longrightarrow &\F_{2^k}\setminus\{1\}\\
\gamma\,\,\,\,&\mapsto &1+\dfrac{a\beta^{2^{k-1}(2^s+1)+1}}{\gamma}.
\end{array}
$$
Since $a\beta^{2^{k-1}(2^s+1)+1}\in\F_{2^k}^*$,
the map $\Upsilon$ is a bijection from
$\F_{2^k}^*$ onto $\F_{2^k}\setminus\{1\}$,
with inverse
\[
y\mapsto
\frac{a\beta^{2^{k-1}(2^s+1)+1}}{y+1}.
\]
Since $\Tr_r^k(1)=1$, we have $1\notin\left\{v^{2^s}+v:v\in V\right\}$ by Lemma \ref{tr}.
From the bijectivity of $\Upsilon$ and \eqref{car}, the equation
\eqref{vr} in variable $v$ has exactly $2^r$ solutions in $V$ for precisely
$2^{k-r}$ values of $\gamma\in\F_{2^k}^{*}$.

By Proposition \ref{pro1} and \eqref{bb},
\[
\begin{aligned}
	b=\bar a\beta^{-1}+\gamma\beta^{2^{k-1}(2^s+1)}=0
	&\Longleftrightarrow
	\gamma=\bar a\beta^{-2^{k-1}(2^s+1)-1}.
\end{aligned}
\]
Since $a\beta^{2^{k-1}(2^s+1)+1} \in \F_{2^k}^*$, $\bar a\beta^{-2^{k-1}(2^s+1)-1}
=a\beta^{2^{k-1}(2^s+1)+1}$. Then in this case the equation  \eqref{vr} reduces to \[	v^{2^s}+v=1+\frac{a\beta^{2^{k-1}(2^s+1)+1}}{\gamma}=1+\frac{	a\beta^{2^{k-1}(2^s+1)+1}}{\bar a\beta^{-2^{k-1}(2^s+1)-1}}=0.\]
By Lemma~\ref{tr}, the equation $v^{2^s}+v=0$ has exactly $2^r$ solutions in $V$. Hence this value of $\gamma$ corresponding to $b=0$ is included in the preceding count of $2^{k-r}$ values. However, $b=0$ should be excluded since $B_{2^r+1}(a,\beta)$ counts only nonzero $b$. Consequently, we have
\[
B_{2^r+1}(a,\beta)=2^{k-r}-1=\frac{2^k-2^r}{2^r}.
\]

\medskip \noindent\emph{Case 2: $\bar a\beta^{2^s}+a\beta\ne0$.}

In this case, using \eqref{xaddeq1}, we obtain
$$
\Tr_k^n\left(
\frac{a\beta}{\bar a\beta^{2^s}+a\beta}
\right)
=
\frac{
\Tr_k^n\left(a\beta^{2^{k-1}(2^s+1)+1}\right)
}{
\beta^{2^{k-1}(2^s+1)}
(\bar a\beta^{2^s}+a\beta)
}=1.$$
Thus,
$\dfrac{a\beta}{\bar a\beta^{2^s}+a\beta}\in V$, where the set $V$ is defined in \eqref{V}. Further, by Lemma \ref{tr}, we have
$\left(\frac{a\beta}{\bar a\beta^{2^s}+a\beta}\right)^{2^s}+\frac{a\beta}{\bar a\beta^{2^s}+a\beta}\in \F_{2^k}$ and ${\rm Tr}_r^k\left(\left(\frac{a\beta}{\bar a\beta^{2^s}+a\beta}
\right)^{2^s}+\frac{a\beta}{\bar a\beta^{2^s}+a\beta}\right)=0$.

Since $\dfrac{a\beta}{\bar a\beta^{2^s}+a\beta}\in V$,
the definition of $V$ shows that every $v\in V$
can be written uniquely as
\begin{equation}\label{v}
v=u+\frac{a\beta}{\bar a\beta^{2^s}+a\beta},
\quad u\in\F_{2^k}.
\end{equation}

Set $\delta=\left(\frac{a\beta}{\bar a\beta^{2^s}+a\beta}\right)^{2^s}
		+\frac{a\beta}{\bar a\beta^{2^s}+a\beta}$. Substituting \eqref{v} into \eqref{affine} gives
\begin{equation}\label{*}
		\begin{array}{lcl}
		&&\biggl[\beta^{2^{k-1}(2^s+1)}\bigl(\bar a\beta^{2^s}+a\beta\bigr)+\gamma\biggr]
		\biggl(u+\frac{a\beta}{\bar a\beta^{2^s}+a\beta}\biggr)^{2^s}
		+\gamma \biggl(u+\frac{a\beta}{\bar a\beta^{2^s}+a\beta}\biggr)\\
		&&\,\,\,\,\,\,\,\,\,\,\,\,\,\,\,\,\,\,\,\,\,\,\,\,\,\,\,\,\,\,\,\,\,\,\,\,\,\,\,\,\,\,\,\,\,\,
\,\,\,\,\,\,\,\,\,\,\,\,\,\,\,\,\,\,\,\,\,\,\,\,\,\,\,\,\,\,\,\,\,\,\,\,\,\,\,\,\,\,\,\,\,\,\,\,\,\,\,\,
\,\,\,\,\,\,\,\,\,\,\,\,\,\,\,\,\,\,\,\,\,\,\,\,\,\,\,\,\,\,+a\beta^{2^{k-1}(2^s+1)+1}+\gamma=0 \\
		\Leftrightarrow
		&&\biggl[\beta^{2^{k-1}(2^s+1)}\bigl(\bar a\beta^{2^s}+a\beta\bigr)+\gamma\biggr]u^{2^s}+\gamma u
		\\
		&&\,\,\,\,\,\,\,\,\,\,\,\,\,\,\,\,\,\,\,\,\,\,\,\,\,\,\,\,\,\,\,\,\,\,\,\,\,\,\,\,\,\,\,\,\,\,\,\,\,\,\,\,\,\,\,\,\,\,\,\,\,\,\,\,\,\,\,=\beta^{2^{k-1}(2^s+1)}\bigl(\bar a\beta^{2^s}+a\beta\bigr)
		\delta+\gamma(\delta+1) \\
		\Leftrightarrow
		&&\biggl[\beta^{2^{k-1}(2^s+1)}\bigl(\bar a\beta^{2^s}+a\beta\bigr)+\gamma\biggr]u^{2^s}+\gamma u
		=\biggl(\beta^{2^{k-1}(2^s+1)}\bigl(\bar a\beta^{2^s}+a\beta\bigr)+\gamma\biggr)\delta+\gamma.
	\end{array}
\end{equation} By \eqref{xaddeq1}, we have
\[
\beta^{2^{k-1}(2^s+1)}
\bigl(\bar a\beta^{2^s}+a\beta\bigr)\in\F_{2^k}.
\]
Since $\delta,\gamma\in\F_{2^k}$, all coefficients
of the equation \eqref{*} in $u$ belong to $\F_{2^k}$. We now count the parameters $\gamma\in\F_{2^k}$
such that the equation \eqref{*} has exactly $2^r$
solutions in $\F_{2^k}$.

When
$
\gamma=\beta^{2^{k-1}(2^s+1)}\left(\bar a\beta^{2^s}+a\beta	\right),
$
the equation \eqref{*} reduces to
\[
\gamma u=\gamma,
\]
and admits the unique solution $u=1$ since $\gamma\neq 0$. Thus, the element $b$ corresponding to this value
of $\gamma$ is not counted by $B_{2^r+1}(a,\beta)$.

When $\gamma\ne \beta^{2^{k-1}(2^s+1)}\left(\bar a\beta^{2^s}+a\beta\right)$, the equation \eqref{*} becomes
\begin{equation}\label{ne}
	u^{2^s}+\frac{\gamma}{\beta^{2^{k-1}(2^s+1)}\left(\bar a\beta^{2^s}+a\beta\right)+\gamma}u
	=
	\delta+\frac{\gamma}{\beta^{2^{k-1}(2^s+1)}\left(\bar a\beta^{2^s}+a\beta\right)+\gamma}.
\end{equation}
Consider the equation
\begin{equation}\label{u}
	u^{2^s}+\frac{\gamma}{\beta^{2^{k-1}(2^s+1)}\left(\bar a\beta^{2^s}+a\beta\right)+\gamma}u=0.
\end{equation}
To ensure that \eqref{ne} has exactly $2^r$
solutions in $\F_{2^k}$, the homogeneous equation
\eqref{u} must have exactly $2^r$ solutions,
and \eqref{ne} must have a solution in $\F_{2^k}$.
When $u\ne0$, the equation \eqref{u} is equivalent to $u^{2^s-1}=\frac{\gamma}{\beta^{2^{k-1}(2^s+1)}\left(\bar a\beta^{2^s}+a\beta\right)+\gamma}$. Since $\gcd(2^s-1,2^k-1)=2^r-1$ by \eqref{3.1},
the equation \eqref{u} has $2^r$ solutions in $\mathbb{F}_{2^k}$ precisely if
$
\frac{\gamma}{\beta^{2^{k-1}(2^s+1)}\left(\bar a\beta^{2^s}+a\beta\right)+\gamma}
$
is a $(2^s-1)$-th power in $\mathbb{F}_{2^k}^*$, that is to say,
\[
\frac{\gamma}{\beta^{2^{k-1}(2^s+1)}\left(\bar a\beta^{2^s}+a\beta\right)+\gamma}=\eta^{2^s-1}
\]
for some $\eta\in\mathbb{F}_{2^k}^{*}$. Using this relation and substituting $u=x\eta$ with $x \in \mathbb{F}_{2^k}$ into \eqref{ne}, we obtain
\[x^{2^s}+x=\delta \eta^{-2^s}+\eta^{-1}=\delta (\eta^{-1})^{2^s}+\eta^{-1}, \]
where $\delta=\left(\frac{a\beta}{\bar a\beta^{2^s}+a\beta}\right)^{2^s}
		+\frac{a\beta}{\bar a\beta^{2^s}+a\beta}$.
Since $\eta\ne0$, this substitution preserves
the number of solutions in $\F_{2^k}$. By Lemma \ref{tr}, the resulting equation
has exactly $2^r$ solutions in $\F_{2^k}$
if and only if
\begin{equation} \label{trace}
\Tr_r^k\left(\delta\eta^{-2^s}+\eta^{-1}\right)=0.
\end{equation}
Since $\Tr_r^k(\delta)=0$, Corollary 1 shows that exactly $2^{k-r}-1$ elements
$\eta\in\F_{2^k}^*$ satisfy  \eqref{trace}.

Every such $\eta$ satisfies $\eta^{2^s-1}\ne1$.
Otherwise $\eta\in\F_{2^r}^*$ and the trace in \eqref{trace}
would equal $\eta^{-1}(\Tr_r^k(\delta)+\Tr_r^k(1))=\eta^{-1}\ne0$,
since $k/r$ is odd. Thus, each such $\eta$ determines a unique parameter
\[
\gamma=
\frac{\eta^{2^s-1}}{1+\eta^{2^s-1}}\,
\beta^{2^{k-1}(2^s+1)}
(\bar a\beta^{2^s}+a\beta).
\]In particular, this value of $\gamma$ is never
equal to
$\beta^{2^{k-1}(2^s+1)}(\bar a\beta^{2^s}+a\beta)$,
since $\eta^{2^s-1}/(1+\eta^{2^s-1})\ne1$.
These elements $\eta$ give all the parameters $\gamma$
such that \eqref{ne} has exactly $2^r$ solutions in $\F_{2^k}$.

It remains to find how many distinct elements
$\eta$ result in the same parameter $\gamma$. Two elements give the same $\gamma$ if and only
if they give the same power value $\eta^{2^s-1}$. Let $\eta_0$ be an element satisfying the trace in \eqref{trace}.
By $\gcd(2^s-1,2^k-1)=2^r-1$, the elements
giving the same power value as $\eta_0$
are precisely $\eta_0\mu$ with
$\mu\in\F_{2^r}^*$. For every such $\mu$, we have $\mu^{2^s}=\mu$,
and hence
\[
\Tr_r^k\left(
\delta(\eta_0\mu)^{-2^s}+(\eta_0\mu)^{-1}
\right)
=
\mu^{-1}\Tr_r^k\left(
\delta\eta_0^{-2^s}+\eta_0^{-1}
\right)=0.
\]
Thus, each parameter $\gamma$ counted above
corresponds to exactly $2^r-1$ of the
$2^{k-r}-1$ elements $\eta$ satisfying
\eqref{trace}.
The number of these parameters $\gamma$
is therefore $$\frac{2^{k-r}-1}{2^r-1}.$$

Finally, for every parameter counted above,
$\beta$ is a solution of \eqref{3.61}.
Setting $b=0$ in this equation would give $\bar a\beta^{2^s}+a\beta=0$,
contradicting the assumption of Case 2.
Thus, the corresponding elements $b$ are
all nonzero.
Since $b$ and $\gamma$ are in bijection
for fixed $a$ and $\beta$, we obtain
\[
B_{2^r+1}(a,\beta)
=\frac{2^{k-r}-1}{2^r-1}
=\frac{2^k-2^r}{2^r(2^r-1)}.
\] \end{proof}

\begin{remark}\label{adrem1}
For $i\in\{1,2\}$, define
\[
B_i(a,\beta)
=
\bigl|\{b\in\F_{2^n}^{*}:N(a,b)=i
\text{ and }\beta\text{ is a root of }\eqref{3.61}\}\bigr|.
\]
Proposition~\ref{pro1} also allows us to determine these
two quantities. Indeed, $N(a,b)=1$ or $2$ if and
only if the corresponding equation in $V$ has
no solution or exactly one solution, respectively.
Following the two cases in the proof of
Proposition \ref{lem5}, we obtain
\[
B_1(a,\beta)=
\begin{cases}
2^k-2^{k-r},
& \bar a\beta^{2^s}+a\beta=0,\\
2^{k-r}-1,
& \bar a\beta^{2^s}+a\beta\ne0
\end{cases}
\]
and
\[
B_2(a,\beta)=
\begin{cases}
0,
& \bar a\beta^{2^s}+a\beta=0,\\[1ex]
\displaystyle
2^k+1-\frac{2^k-1}{2^r-1},
& \bar a\beta^{2^s}+a\beta\ne0.
\end{cases}
\]
\end{remark}

By Proposition \ref{lem5} and the double-counting identity \eqref{AB},
to determine $A_{2^r+1}(a)$, it remains to count
the elements $\beta\in S$ satisfying
$\bar a\beta^{2^s}+a\beta=0$.

\begin{thm}\label{prop2}
	Let $A_{2^r+1}(a)$ be defined in \eqref{Ai} and $\alpha$ be a primitive element
of $\F_{2^n}$. For each integer $l$ with $0\le l\le2^r$,
set $a=\alpha^l$. Then
	\begin{equation}\label{3.17}
		A_{2^r+1}(a)=\left\{
		\begin{array}{cc}
			\frac{(2^k-2^r)(2^{2r}-1+2^k-2^r)}
			{2^r(2^{2r}-1)},&l=0,\\
			\frac{(2^k+1)(2^k-2^r)}
			{2^r(2^{2r}-1)},&1\leq l\leq 2^r.
		\end{array}\right.
	\end{equation}
\end{thm}

\begin{proof}
	Since $a\ne0$ and $\beta\in S$, the equality $\bar a\beta^{2^s}+a\beta=0$ is equivalent to \begin{equation}\label{ba}
		\beta^{2^s-1}=a^{1-2^k}.
	\end{equation}
The element $\alpha^{2^k-1}$ generates the cyclic group $S$, so every $\beta\in S$ can be written uniquely as $\beta=(\alpha^{2^k-1})^j$ for some $j\in \{0,1,2,\cdots,2^k\}$. Using $a=\alpha^l$, we obtain
\[
\bar a\beta^{2^s}+a\beta=0
\quad\Longleftrightarrow\quad
(2^s-1)j\equiv-l\pmod{2^k+1}.
\]By $\gcd(2^s-1,2^k+1)=2^r+1$ in \eqref{3.1},
this congruence has a solution if and only if
$2^r+1$ divides $l$. Since $0\le l\le2^r$, this occurs precisely
when $l=0$.
In that case, there are exactly $2^r+1$
solutions modulo $2^k+1$. Hence
$$\big|\{\beta \in S: \bar a\beta^{2^s}+a\beta=0\}\big|=\left\{
	\begin{array}{cc}
		2^r+1, & l=0,\\
		0, & 1\leq l\leq 2^r.
	\end{array}\right.$$
Since $|S|=2^k+1$, the number of elements
$\beta\in S$ satisfying
$\bar a\beta^{2^s}+a\beta\ne0$ is
$2^k-2^r$ when $l=0$ and $2^k+1$
when $1\le l\le2^r$.
	
	Therefore, by Proposition \ref{lem5}, we have
	\begin{align*}
		\sum\limits_{\beta\in S}B_{2^r+1}(a,\beta)
		&=
		\sum_{\substack{\beta\in S\\\bar a\beta^{2^s}+a\beta=0}}B_{2^r+1}(a,\beta)
		+
		\sum_{\substack{\beta\in S\\\bar a\beta^{2^s}+a\beta\ne0}}B_{2^r+1}(a,\beta)
		\nonumber\\
		&=\left\{
		\begin{array}{cc}
			(2^r+1)\frac{2^k-2^r}{2^r}
			+(2^k-2^r)\frac{2^k-2^r}{2^r(2^r-1)},
			& l=0,\\
			(2^k+1)\frac{2^k-2^r}{2^r(2^r-1)},
			& 1\leq l\leq 2^r.
		\end{array}\right.
	\end{align*}
This together with the double-counting identity
\eqref{AB} yields the desired formulas.
	\end{proof}

Let $d$ be defined in \eqref{new-d} and $\alpha$ be a primitive element of $\F_{2^n}$.
For each fixed integer $l$ with $0\le l\le2^r$,
set $a=\alpha^l$.
By \eqref{CC}, we have
\[
C_{d,l}(\tau)=W_d(a,\alpha^\tau)-1.
\]
As $\tau$ ranges over $0\le\tau\le2^n-2$,
the element $\alpha^\tau$ runs through
$\F_{2^n}^*$.
By Proposition~\ref{lem 4} and the definition of $A_i(a)$,
the correlation value $2^k(i-1)-1$ occurs with frequency
$A_i(a)$ for each $i\in\{0,1,2,2^r+1\}$.
Since $A_{2^r+1}(a)$ has been determined in
Theorem~\ref{prop2}, it remains to compute
$A_0(a)$, $A_1(a)$, and $A_2(a)$.

By Remark~\ref{adrem1} and the double-counting identities
\[
iA_i(a)=\sum_{\beta\in S}B_i(a,\beta),
\qquad i=1,2,
\]
we can determine $A_1(a)$ and $A_2(a)$ in the same way
as in the proof of Theorem~\ref{prop2}.
The total frequency identity
\begin{equation}\label{sumeq3}
A_0(a)+A_1(a)+A_2(a)+A_{2^r+1}(a)=2^n-1
\end{equation}
then gives $A_0(a)$, and hence the correlation distribution
for each $l\in\{0,1,\ldots,2^r\}$.

Here we give an alternative derivation using the first
two moments. Together with \eqref{sumeq3}, the following
two moment identities give three linear equations
for the remaining frequencies:
\begin{align}
\sum_{\tau=0}^{2^n-2} C_{d,l}(\tau)
&=
\sum_{i\in\{0,1,2,2^r+1\}}
\bigl(2^k(i-1)-1\bigr)A_i(a),
\label{first-moment1}\\
\sum_{\tau=0}^{2^n-2} C_{d,l}(\tau)^2
&=
\sum_{i\in\{0,1,2,2^r+1\}}
\bigl(2^k(i-1)-1\bigr)^2A_i(a).
\label{first-moment2}
\end{align}

Following the approach in
\cite[Lemma 2.5]{Niho1}, we can evaluate these
two moments in the following proposition.

\begin{prop}\label{Ca}
Let $C_{d,l}(\tau)$ be defined in \eqref{2.3} and $d$ be given in \eqref{new-d}. For each $l$ with $0\le l\le2^r$,
the first two moments are given by
	\begin{enumerate}
		\item[(i)]$
		\sum\limits_{\tau=0}^{2^n-2} C_{d, l}(\tau)=
			\left\{
		\begin{array}{cc}
			-2^{r+k}+1,&l=0,\\
			2^k+1,&1\leq l\leq 2^r;
		\end{array}\right.$
		\item[(ii)]$
		\sum\limits_{\tau=0}^{2^n-2}C_{d, l}^2(\tau)=
		\left\{
		\begin{array}{cc}
			2^{4k}-2^{2k}-2^{2k+2r}+2^{k+r+1}-1,&l=0,\\
			2^{4k}-2^{2k+1}-2^{k+1}-1,&1\leq l\leq 2^r.
		\end{array}\right.
		$
	
	\end{enumerate}
\end{prop}

\begin{proof} (i) Put $a=\alpha^l$ and $b=\alpha^{\tau}$. As $\tau$ ranges over
$\{0,1,\ldots,2^n-2\}$, the element $b=\alpha^\tau$
runs through $\F_{2^n}^{*}$.
By \eqref{2.3}, we have
	\begin{equation}\label{powerad2}
	\begin{aligned}
		\sum_{\tau=0}^{2^n-2} C_{d, l}(\tau)
		&=\sum_{b\in\mathbb{F}_{2^n}^*}\sum_{x\in\mathbb{F}_{2^n}^*}(-1)^{\operatorname{Tr}_1^n(ax^d+bx)}\\
		&=\sum_{x\in\mathbb{F}_{2^n}^*}(-1)^{\operatorname{Tr}_1^n(ax^d)}\sum_{b\in\mathbb{F}_{2^n}^*}(-1)^{\operatorname{Tr}_1^n(bx)}\\
		&=\sum_{x\in\mathbb{F}_{2^n}^*}(-1)^{\operatorname{Tr}_1^n(ax^d)}\cdot (-1)\\
		&=-\sum_{x\in\mathbb{F}_{2^n}^*}(-1)^{\operatorname{Tr}_1^n(ax^d)}.
	\end{aligned}
	\end{equation}
Since $\gcd(d,2^n-1)=2^r+1$, the cyclic subgroups
$\langle\alpha^d\rangle$ and
$\langle\alpha^{2^r+1}\rangle$ both have order
$(2^n-1)/(2^r+1)$ and hence coincide.
As $i$ ranges over $\{0,1,\ldots,2^n-2\}$, both
$\alpha^{id}$ and $\alpha^{i(2^r+1)}$ run through
each element of this common subgroup exactly
$2^r+1$ times. Therefore,
\begin{equation}\label{powerad1}
\begin{aligned}
\sum_{x\in\F_{2^n}^{*}}(-1)^{\Tr_1^n(ax^d)}
&=\sum_{i=0}^{2^n-2}
  (-1)^{\Tr_1^n(a\alpha^{id})}\\
&=\sum_{i=0}^{2^n-2}
  (-1)^{\Tr_1^n(a\alpha^{i(2^r+1)})}\\
&=\sum_{x\in\F_{2^n}^{*}}
  (-1)^{\Tr_1^n(ax^{2^r+1})}.
\end{aligned}
\end{equation}

Recall that $r\mid k$ and $k/r$ is odd.
The least positive integer $j$ satisfying
$2^r+1\mid 2^j+1$ is $j=r$. By Lemma \ref{lem 2},
\begin{equation}\label{power}
\sum_{x\in\mathbb F_{2^n}}(-1)^{\mathrm{Tr}_1^n(ax^{2^r+1})}=
\begin{cases}
	2^{r+k}, & \text{if } a \text{ is a }(2^r+1)\text{-th power in } \mathbb{F}_{2^n}^{*},\\
	-2^{k}, & \text{if } a \text{ is not a }(2^r+1)\text{-th power in } \mathbb{F}_{2^n}^{*}.
\end{cases}
\end{equation}
Since $a=\alpha^l$ with $0\le l\le 2^r$,
we have $a\in\langle\alpha^{2^r+1}\rangle$
if and only if $l=0$. Removing the term corresponding
to $x=0$ and combining
\eqref{powerad2}, \eqref{powerad1}, and \eqref{power}
yield the desired formula in (i).

 (ii) Expanding the square, we have
\begin{align*}
	\sum_{\tau=0}^{2^n-2} C_{d, l}^2(\tau)
	&=\sum_{b\in\mathbb{F}_{2^n}^*}\sum_{x_1\in\mathbb{F}_{2^n}^*}(-1)^{\mathrm{Tr}_1^n(ax_1^d+bx_1)}\sum_{x_2\in\mathbb{F}_{2^n}^*}(-1)^{\mathrm{Tr}_1^n(ax_2^d+bx_2)}\\
	&= \sum_{x_1\in\mathbb{F}_{2^n}^*}\sum_{x_2\in\mathbb{F}_{2^n}^*}
	(-1)^{\mathrm{Tr}_1^n\bigl(a\left(x_1^d+x_2^d\right)\bigr)}
	\sum_{b\in\mathbb{F}_{2^n}^*}
	(-1)^{\mathrm{Tr}_1^n(b(x_1+x_2))}.
\end{align*}

Since \[\sum\limits_{b\in\mathbb{F}_{2^n}^*}(-1)^{\mathrm{Tr}_1^n(b(x_1+x_2))}=\left\{\begin{array}{cc}2^n-1, &  x_1=x_2,\\-1,& x_1\ne x_2,\end{array}\right.\]
 and $x_1^d+x_2^d=0$ when $x_1=x_2$, it follows that
\[
\begin{aligned}
\sum_{\tau=0}^{2^n-2}C_{d,l}(\tau)^2
&=(2^n-1)^2
  -\sum_{\substack{x_1,x_2\in\F_{2^n}^{*}\\x_1\ne x_2}}
  (-1)^{\Tr_1^n\left(a\left(x_1^d+x_2^d\right)\right)}\\
&=(2^n-1)^2+(2^n-1)
  -\sum_{x_1,x_2\in\F_{2^n}^{*}}
  (-1)^{\Tr_1^n\left(a\left(x_1^d+x_2^d\right)\right)}\\
&=2^n(2^n-1)
  -\left(
    \sum_{x\in\F_{2^n}^{*}}
    (-1)^{\Tr_1^n(ax^d)}
   \right)^2.
\end{aligned}
\]

By the evaluation in (i),
\[
\sum_{x\in\F_{2^n}^{*}}(-1)^{\Tr_1^n(ax^d)}
=
\begin{cases}
2^{k+r}-1, & l=0,\\
-2^k-1, & 1\le l\le 2^r.
\end{cases}
\]

Substituting these values and using $n=2k$, we obtain
the formula in (ii).
\end{proof}

With the above preparations, we can compute the cross-correlation distribution between $\{u(t)\}$ and its $d$-decimated sequences $\{v_l(t)\}$ for each $l\in \{0,1,2,\cdots,2^r\}$.

\begin{thm}\label{thm} Let $n=2k$, $d$ be given in \eqref{new-d}, and $r=\gcd(s,k)$. Then, for each fixed integer $l$ in $\{0,1,2,\cdots,2^r\}$, the value distribution of $C_{d,l}(\tau)$ as
$\tau$ ranges over $\{0,1,\ldots,2^n-2\}$ is given
as follows.

	(i) If $l=0$, then
	\[
	\left\{
	\begin{array}{cc}
		-1-2^k, & \frac{2^r\left(2^{2k}-2^{k+r}+2^k-2^r\right)}{2(2^r+1)} \text{ times},\\[6pt]
		-1, & 2^{2k-r}+2^{k+r}-2^{k+1}-2^{k-r}+2^r \text{ times},\\[6pt]
		-1+2^k, & \frac{(2^{k}-2^{r})(2^{k+r}-2^{k+1}+2^{r})}{2(2^r-1)} \text{ times},\\[10pt]
		-1+2^{k+r}, & \frac{(2^k-2^r)\left(2^k-2^r+2^{2r}-1\right)}{2^r(2^{2r}-1)} \text{ times}.
	\end{array}\right.
	\]
	
	(ii) For each integer $l$ with $1\le l\le 2^r$, the distribution is given by
	\[
	C_{d, l}(\tau)=
	\left\{
	\begin{array}{cc}
		-1-2^k, & \frac{2^{2k+r}-2^{k+1}-2^r-2}{2(2^r+1)} \text{ times},\\[6pt]
		-1, & 2^{2k-r}-2^k+2^{k-r}-1 \text{ times},\\[6pt]
		-1+2^k, & \frac{(2^k+1)\left(2^{k+r}-2^{k+1}+2^r\right)}{2(2^r-1)} \text{ times},\\[10pt]
		-1+2^{k+r}, & \frac{(2^k+1)(2^k-2^r)}{2^r(2^{2r}-1)} \text{ times}.
	\end{array}\right.
	\]
\end{thm}

\begin{proof}
For each fixed $l\in\{0,1,\ldots,2^r\}$, combining
the identities \eqref{sumeq3}--\eqref{first-moment2}
with the moment formulas in Proposition~\ref{Ca}
and the expression for $A_{2^r+1}(\alpha^l)$ in
Theorem~\ref{prop2} gives the following systems.

(i) If $l=0$, then
\[
	\left\{
	\begin{aligned}
		&A_0(a)+A_1(a)+A_2(a)+A_{2^r+1}(a)=2^{n}-1,\\
		&(-1-2^k)A_0(a)-A_1(a)+(-1+2^k)A_2(a)
		+(-1+2^{k+r})A_{2^r+1}(a)=-2^{r+k}+1,\\
		&(-1-2^k)^2A_0(a)+A_1(a)+(-1+2^k)^2A_2(a) +(-1+2^{k+r})^2A_{2^r+1}(a)\\
		&\qquad\,\,\,\,\,\,\,\,\,\,\,\,\,\,\,\,\,\,\,\,\,\,\,\,\,\,\,\,\,\,\,\,\,\,\,\,\,\,\,\,\,\,\,\,\,\,\,\,\,\,\,\,\,\,\,\,\,\,\,\,\,\,\,\,\,\,\,\,\,\,\,\,\,\,\,\,\,\,\,\,\,\,\,\,\,\,\,\,\,\,\,\,\,\,\,\,\,\,\,\,\,\,\, =2^{4k}-2^{2k}-2^{2k+2r}+2^{k+r+1}-1,\\
		&A_{2^r+1}(a)=\frac{(2^k-2^r)(2^k-2^r+2^{2r}-1)}{2^r(2^{2r}-1)}.
	\end{aligned}
	\right.
	\]
Solving this system yields the frequencies in (i).

(ii) If $1\le l\le 2^r$, then
\[
	\left\{
	\begin{aligned}
		&A_0(a)+A_1(a)+A_2(a)+A_{2^r+1}(a)=2^{n}-1,\\
		&(-1-2^k)A_0(a)-A_1(a)+(-1+2^k)A_2(a)
		+(-1+2^{k+r})A_{2^r+1}(a)=2^k+1,\\
		&(-1-2^k)^2A_0(a)+A_1(a)+(-1+2^k)^2A_2(a) +(-1+2^{k+r})^2A_{2^r+1}(a)\\
		&\qquad\,\,\,\,\,\,\,\,\,\,\,\,\,\,\,\,\,\,\,\,\,\,\,\,\,\,\,\,\,\,\,\,\,\,\,\,\,\,\,\,\,\,\,\,\,\,\,\,\,\,\,\,\,\,\,\,\,\,\,\,\,\,\,\,\,\,
\,\,\,\,\,\,\,\,\,\,\,\,\,\,\,\,\,\,\,\,\,\,\,\,\,\,\,\,\,\,\,\,\,\,\,\,
\,\,\,\,\,\,\,\,\,\,\,\,\, =2^{4k}-2^{2k+1}-2^{k+1}-1,\\
		&A_{2^r+1}(a)=\frac{(2^k+1)(2^k-2^r)}
		{2^r(2^{2r}-1)}.
	\end{aligned}
	\right.
\]
Solving this system gives the frequencies in (ii).
\end{proof}

The following numerical examples verify the distributions
given by Theorem \ref{thm}.

\begin{example}
For the parameter pairs $(s,k)=(2,3)$, $(4,6)$, and $(6,9)$
in \eqref{new-d}, the values of $C_{d,l}(\tau)$
and their frequencies computed using Magma are listed in Table~\ref{exam01}.
They are consistent with the results in Theorem \ref{thm}.
\end{example}

\begin{table}[H]
	\caption{Some numerical results about the correlation distributions}\vspace{0.2cm}
	\label{exam01}
	\centering
	\small
	\setlength{\tabcolsep}{5.5pt}
	\renewcommand{\arraystretch}{1.08}
	\begin{tabular}{ccccccc}
		\toprule
		\multirow{2}{*}{\((s,k)\)}
		& \multirow{2}{*}{\(d\)}
		& \multirow{2}{*}{\(r\)}
		& \multirow{2}{*}{\(2^r+1\)}
		& \multirow{2}{*}{\(\begin{gathered}
				\text{Cross-correlation}\\[-2pt]
				\text{value}\,\,\, C_{d,l}
			\end{gathered}\)}
		& \multicolumn{2}{c}{Frequency}
		\\
		\cmidrule(lr){6-7}
		&&&&
		& \(\begin{gathered}
						l=0
		\end{gathered}\)
		& \(\begin{gathered}
						l\in\{1,2,\cdots,2^r\}
		\end{gathered}\)
		\\
		\midrule
		
		\multirow{4}{*}{\((2,3)\)}
		& \multirow{4}{*}{\(57\)}
		& \multirow{4}{*}{\(1\)}
		& \multirow{4}{*}{\(3\)}
		& \(-9\) & \(18\) & \(18\) \\
		&&&& \(-1\) & \(30\) & \(27\) \\
		&&&& \(7\)  & \(6\)  & \(9\)  \\
		&&&& \(15\) & \(9\)  & \(9\)  \\
		\midrule
		
		\multirow{4}{*}{\((4,6)\)}
		& \multirow{4}{*}{\(2710\)}
		& \multirow{4}{*}{\(2\)}
		& \multirow{4}{*}{\(5\)}
		& \(-65\) & \(1560\) & \(1625\) \\
		&&&& \(-1\)  & \(1140\) & \(975\)  \\
		&&&& \(63\)  & \(1320\) & \(1430\) \\
		&&&& \(255\) & \(75\)   & \(65\)   \\
		\midrule
		
		\multirow{4}{*}{\((6,9)\)}
		& \multirow{4}{*}{\(149724\)}
		& \multirow{4}{*}{\(3\)}
		& \multirow{4}{*}{\(9\)}
		& \(-513\) & \(114912\) & \(116451\) \\
		&&&& \(-1\)   & \(35784\)  & \(32319\)  \\
		&&&& \(511\)  & \(110880\) & \(112860\) \\
		&&&& \(4095\) & \(567\)    & \(513\)    \\
		\bottomrule
	\end{tabular}

\end{table}

\begin{remark}
For each fixed $k$, Theorem~\ref{thm} treats the two
noncoprime branches of \eqref{d12} in the unified form
\eqref{new-d}, with $1<s<2k$ and $v_2(s)>v_2(k)$. Since $r<k$ and $\frac{k}{r}$ is odd, we have $k\ge3r$.
The four frequencies in each case are therefore
positive, so $C_{d,l}(\tau)$ takes exactly four values
for every $0\le l\le2^r$.
\end{remark}

\begin{remark}
Theorem~\ref{thm} includes the binary cross-correlation
functions studied in \cite{Fan1,Luo1,Xia2} as special cases.
As noted after \eqref{forcomp}, the corresponding
decimations are obtained from \eqref{new-d} by taking
$s=k+t$, $s=2k-2$, and $s=k+1$, respectively.
Under the parameter conditions in these references,
$\gcd(s,k)=1$ and $\gcd(d,2^n-1)=3$.
Thus, these results correspond to the case $r=1$
of Theorem~\ref{thm}.
\end{remark}

We conclude this section by comparing our problem
and method with those in the existing literature.
A standard method for determining a correlation
distribution with known possible values is to use
moment identities to obtain linear equations
in the unknown frequencies.
When $\gcd(d,2^n-1)=1$, this method is effective
if $C_{d,0}(\tau)$ takes few values; see
\cite{Helle1,HelleOpen1,Helle3,Helle4,Ranto1,Rose1}.

When $\gcd(d,2^n-1)>1$, the first two moments of
$C_{d,l}(\tau)$ can be evaluated once $W_d(\alpha^l,0)$
is known, but higher moments generally require
additional counting arguments.
The moment identities for $W_d(a,b)$ remain valid,
but they involve the frequencies from all the
distributions of $W_d(\alpha^l,b)$
(equivalently, $C_{d,l}(\tau)$),
$0\le l<\gcd(d,2^n-1)$, and do not, in general,
determine each distribution separately.
Thus, even when the overall distribution of $W_d(a,b)$
is known, additional information is generally needed
to determine the correlation distribution
for each fixed $l$.

In \cite{Fan1,Luo1,Xia2}, the value distribution of
$W_d(a,b)$ is first determined using moment identities up to
order three, and additional information is used to
separate the distributions for $l=0,1,2$.
In \cite{Fan1,Luo1}, where $n=2m$, special choices
of the coefficients give lower bounds for the frequency
of the correlation value $2^m-1$ in each of the two
distribution cases. A relation from the overall
distribution forces both bounds to be attained.
In \cite{Xia2}, a one-to-one correspondence between
a repeated root of an associated cubic and its
coefficient directly gives the frequency of the
correlation value $q-1$ for $l=0$, where $q=p^m$.
The distribution for $l=0$ then follows from the
total frequency and the first two moments, while
those for $l=1,2$ are recovered from the overall
distribution. Thus, all three methods rely on
the overall distribution, but obtain the additional
frequency information in different ways.

Our method differs from those in \cite{Fan1,Luo1,Xia2}
in how the frequency $A_{2^r+1}(a)$ is determined. We determine $A_{2^r+1}(a)$ directly
for each fixed $a=\alpha^l$ by relating the coefficient $b$ to a root $\beta\in S$ of the equation \eqref{3.61} and counting the admissible
pairs $(b,\beta)$ in two ways. The remaining frequencies follow from the total
frequency and the first two moments.
Thus, our method requires neither the overall
distribution of $W_d(a,b)$ nor a third-moment
calculation. It would be interesting to explore whether this
counting method can be applied to other correlation
distribution problems, such as those studied
in \cite{XiaOpen1,XiaNiho1}.

\section{Conclusion}\label{con}

In this paper, we completely determine the cross-correlation
distribution between a binary $m$-sequence of period $2^n-1$
and each of its $d$-decimation sequences, where $d$ is defined
in \eqref{new-d}. There are exactly $2^r+1$ such decimation
sequences, each of period $(2^n-1)/(2^r+1)$, where
$r=\gcd(s,k)$. For each $l\in\{0,1,\ldots,2^r\}$,
the cross-correlation function $C_{d,l}(\tau)$ is four-valued.
The corresponding distributions fall into two classes:
one for $l=0$ and the other for $1\le l\le 2^r$. Our results include those in \cite{Fan1,Luo1,Xia2}
as special cases.

The key step is the direct evaluation of $A_{2^r+1}(a)$
for each fixed $a=\alpha^l$.
The main idea is as follows: we fix a root $\beta\in S$
of \eqref{3.61} and reduce the remaining root-counting
problem to counting solutions of an associated
linearized equation in the trace-one set $V$.
By relating the coefficient $b$ to the root $\beta$,
we establish a double-counting identity for the
admissible pairs $(b,\beta)$ and use it to determine
$A_{2^r+1}(a)$.
The remaining frequencies are then obtained from
the total frequency and the first two moments.
In fact, the same idea yields $A_1(a)$
and $A_2(a)$ directly, while $A_0(a)$ follows
from the total frequency.
Thus, the correlation distributions considered
in this paper can also be determined without
using moment identities.
This counting method offers a new approach to
similar correlation distribution problems.
\section*{Acknowledgement}
{The work of X. Zeng was supported by the National Natural Science Foundation of China under Grant  12471492.}

\end{document}